\author{Martin Radloff}
\pdfoutput=1
	\documentclass[a4paper,12pt,oneside]{scrartcl}
	
\usepackage[sectionbib]{natbib}
\usepackage[]{hyperref}
\renewcommand{\theequation}{\thesection\arabic{equation}}
	\usepackage[english]{babel}   
	\usepackage[latin1]{inputenc}
	\usepackage[T1]{fontenc}
	\usepackage{ifthen}
	\usepackage{amsmath}
	\usepackage{amssymb}
	\usepackage{amsthm}
	\usepackage{latexsym, mathrsfs}
	\usepackage{bbm}

	\usepackage{geometry}
	\usepackage{graphicx}
	\usepackage{subfigure}
	\usepackage{xcolor}
\usepackage{pgf,tikz}
\usetikzlibrary{arrows}	
	\usepackage{multirow}
	\usepackage{multicol}
	\usepackage{enumerate, enumitem}	
	\usepackage{url}
	\usepackage{scrpage2}
	\usepackage{suetterl}	
\newfont{\suet}{suet14}
\newfont{\schwell}{schwell}
\DeclareTextFontCommand{\textsuet}{\suet}

\DeclareTextFontCommand{\textschwell}{\schwell}	

	\usepackage{datetime}
	\usepackage{blindtext}
	\defpagestyle{Kopfi}{
		{} {} {\small Martin Radloff, Rainer Schwabe\hfill $D$-optimal designs for logit and probit models on balls
		}
	}{
		{} {} {\hfill\pagemark}
	}

\newcommand{\BIGOP}[1]{\mathop{\mathchoice%
{\raise-0.22em\hbox{\huge $#1$}}%
{\raise-0.05em\hbox{\Large $#1$}}{\hbox{\large $#1$}}{#1}}}

\newcommand{\BIGboxplus}{\mathop{\mathchoice%
{\raise-0.35em\hbox{\huge $\boxplus$}}%
{\raise-0.15em\hbox{\Large $\boxplus$}}{\hbox{\large $\boxplus$}}{\boxplus}}}

\newtheorem{theorem}{Theorem}
\newtheorem{lemma}{Lemma}

\theoremstyle{definition}

\newtheorem{notation}{Notation}

\newcommand{\D}{\mathrm{d}}
\newcommand{\id}{\operatorname{id}}

\usepackage{dsfont}

\begin{document}
	\pagestyle{Kopfi}
	\thispagestyle{empty}
	\vspace*{1.3cm}
	\begin{center}
		{\Large\textbf{Locally $D$-optimal Designs for a Wider Class of Non-linear Models on the $k$-dimensional Ball}\\ \textbf{\large with applications to logit and probit models}}
	\end{center}
	\centerline{Martin Radloff\footnote[2]{corresponding author: Martin Radloff, Institute for Mathematical Stochastics, Otto-von-Guericke-University, PF~4120, 39016~Magdeburg, Germany, \url{martin.radloff@ovgu.de}} and Rainer Schwabe\footnote[3]{Rainer Schwabe, Institute for Mathematical Stochastics, Otto-von-Guericke-University, PF~4120, 39016~Magdeburg, Germany, \url{rainer.schwabe@ovgu.de}}} 
		
	\vspace*{1.3cm}

%

\begin{quotation}
\noindent {\textit{Abstract:}}
In this paper we extend the results of \cite{Radloff:2018}, which could be applied for example to Poisson regression, negative binomial regression and proportional hazard models with censoring, to a wider class of non-linear multiple regression models. This includes the binary response models with logit and probit link besides other. For this class of models we derive (locally) $D$-optimal designs when the design region is a $k$-di\-men\-sion\-al ball. For the corresponding construction we make use of the concept of invariance and equivariance in the context of optimal designs as in our previous paper. In contrast to the former results the designs will not necessarily be exact designs in all cases. Instead approximate designs can appear. These results can be generalized to arbitrary ellipsoidal design regions.

\vspace{9pt}
\par
\end{quotation}\par

\def\thefigure{\arabic{figure}}
\def\thetable{\arabic{table}}

\renewcommand{\theequation}{\thesection.\arabic{equation}}

\fontsize{12}{14pt plus.8pt minus .6pt}\selectfont

\setcounter{equation}{0} 
\section{Introduction}
\label{intro}
	In \cite{Radloff:2018} we found optimal designs for a special class of linear and non-linear models with respect to the $D$-criterion on a $k$-dimensional ball.
	The main result was for (non-linear) multiple regression models, that means the linear predictor is
	\[\boldsymbol{f}(\boldsymbol{x})^\top\boldsymbol{\beta} = \beta_0 + \beta_1 x_1 + \ldots + \beta_k x_k\ .\]
	For this result to hold the one-support-point (or elemental) information matrix should be representable in the form
	\begin{equation*}
		\boldsymbol{M}(\boldsymbol{x},\boldsymbol{\beta})=\lambda\!\left(\boldsymbol{f}(\boldsymbol{x})^\top\boldsymbol{\beta}\right)\boldsymbol{f}(\boldsymbol{x})\boldsymbol{f}(\boldsymbol{x})^\top
	\end{equation*}
	with an intensity (or efficiency) function $\lambda$ which only depends on the value of the linear predictor.
	By using results on equivariance and invariance of \cite{Radloff:2016}, we rotate the design space, the $k$-dimensional unit ball $\mathbb{B}_k$, and the parameter space $\mathbb{R}^{k+1}$ simultaneously in such a way, that the linear predictor of the multiple regression problem collapses to 
		\begin{equation}
			\boldsymbol{f}(\boldsymbol{x})^\top\boldsymbol{\beta} = \beta_0 + \beta_1 x_1 \text{\quad and \quad $\beta_1\geq 0$} .
		\label{eq:simplifiedModel}
		\end{equation}
	So it is possible to reduce that multidimensional problem to a one-dimensional marginal problem. Similar one-dimensional problems have already been investigated, for example in \cite{Konstantinou:2014}.\\	
	In \cite{Radloff:2018} the following four conditions, which can be satisfied by the intensity function $\lambda$, were imposed {(see also \cite{Konstantinou:2014} or \cite{Schmidt:2017})}:
\begin{enumerate}[leftmargin=1cm]
	\item[(A1)] $\lambda$ is positive on $\mathbb{R}$ and twice continuously differentiable.
	\item[(A2)] The first derivative $\lambda^\prime$ is positive on $\mathbb{R}$.
	\item[(A3)] The second derivative $u^{\prime\prime}$ of $u=\frac{1}{\lambda}$ is injective on $\mathbb{R}$.
	\item[(A4)] The function $\frac{\lambda^\prime}{\lambda}$ is non-increasing.
\end{enumerate}
	Poisson regression, negative binomial regression and special proportional hazard models with censoring (see \cite{Schmidt:2017}) fulfill these four conditions.\\
	For a short notation we will use from now on the abbreviation \[q(x_1):=\lambda(\beta_0+\beta_1 x_1)\ .\] For $\beta_1>0$ the properties (A1), (A2), (A3) and (A4) transfer to $q$, respectively, and vice versa. 
	
	In \cite{Radloff:2018} we established the following main result that is reproduced for the readers' convenience.
	
\begin{theorem} \label{Theorem1}
	There is a (locally) $D$-optimal design for the simplified problem \eqref{eq:simplifiedModel} with $\beta_1>0$ and intensity function satisfying (A1)-(A3) that has one support point in\linebreak $(1,0,\ldots,0)^\top$ and the other $k$~support points are the vertices of an arbitrarily rotated, $(k-1)$-dimensional simplex which is maximally inscribed in the intersection of the $k$-dimensional unit ball and a hyperplane with $x_1=x_{12}^\ast$.\\
	For $k\geq 2\ :\ x_{12}^\ast\in(-1,1)$ is solution of 
	\begin{equation*}
		\frac{q^\prime(x_{12}^\ast)}{q(x_{12}^\ast)}=\frac{2\,(1+kx_{12}^\ast)}{k\,(1-x_{12}^{\ast\ 2})}
	\label{eq:L9}
	\end{equation*}
	and for $k=1$\ :\ It is $x_{12}^\ast=-1$ or $x_{12}^\ast\in[-1,1)$ is solution of
	\begin{equation*}
		\frac{q^\prime(x_{12}^\ast)}{q(x_{12}^\ast)}=\frac{2}{1-x_{12}^{\ast}}\ .
	\label{eq:L92}
	\end{equation*}
	In any case, if additionally (A4) is satisfied, the solution $x_{12}^\ast$ is unique.\\ 
	The design is equally weighted with $\frac{1}{k+1}$. 
\end{theorem}

	If $\beta_1=0$ then the design consisting of the equally weighted vertices of a regular simplex inscribed in the unit sphere, the boundary of the design space, is (locally) $D$-optimal. The orientation is arbitrary.
	
	In the present paper we want to transfer the results for example to binary response models with logit or probit link. Here the intensity functions do not satisfy the conditions (A2) and (A3).\\
	The corresponding problem of logit and probit models in one dimension has already been investigated by \cite{Ford:1992} and \cite{Biedermann:2006}.
	
	We will give here a natural extension to higher dimensions.

\setcounter{equation}{0} 
\section{General Model Description, Design, and Invariance}
\label{sec:2}
	In the following sections as mentioned in the introduction we want to focus on a class of (non-linear) multiple regression models. Here every observation $Y$ depends on a special setting of control variables, the design point $\boldsymbol{x}$, which is in the design region $\mathscr{X}=\mathbb{B}_k=\linebreak[2]\{\boldsymbol{x}\in\mathbb{R}^k\ :\ x_1^2+\ldots+x_k^2\leq 1\}$, the $k$-dimensional unit ball with $k\in\mathbb{N}$. The regression function $\boldsymbol{f}:\mathscr{X}\to\mathbb{R}^{k+1}$ is considered to be $\boldsymbol{x}\mapsto(1,x_1,\ldots,x_k)^\top$, and the parameter vector $\boldsymbol{\beta}=(\beta_0,\beta_1,\ldots,\beta_k)^\top$ is unknown and lies in the parameter space $\mathscr{B}$. We will take $\mathscr{B}=\mathbb{R}^{k+1}$.
	So the linear predictor is
	\[\boldsymbol{f}(\boldsymbol{x})^\top\boldsymbol{\beta} = \beta_0 + \beta_1 x_1 + \ldots + \beta_k x_k\ .\] 	 	
	A second requirement is that the one-support-point (or elemental, see \cite{Atkinson:2014}) information matrix $\boldsymbol{M}(\boldsymbol{x},\boldsymbol{\beta})$ can be written as
	\begin{equation*}
		\boldsymbol{M}(\boldsymbol{x},\boldsymbol{\beta})=\lambda\!\left(\boldsymbol{f}(\boldsymbol{x})^\top\boldsymbol{\beta}\right)\boldsymbol{f}(\boldsymbol{x})\boldsymbol{f}(\boldsymbol{x})^\top
	\end{equation*}
	with an intensity (or efficiency) function $\lambda$ (see \citet[Section~1.5]{Fedorov:1972}) which only depends on the value of the linear predictor.	
	
	We want to find optimal designs on the the $k$-dimensional unit ball for those problems. This will be done in the sense of $D$-optimality, which is a very popular criterion and minimizes the volume of the (asymptotic) confidence ellipsoid.\\
	For that account we need the concept of information matrices. In our case the information matrix of a (generalized) design $\xi$ with independent observations is
	\begin{equation*}
		\boldsymbol{M}(\xi,\boldsymbol{\beta})=\int_\mathscr{X}\boldsymbol{M}(\boldsymbol{x},\boldsymbol{\beta})\ \xi(\D \boldsymbol{x})=\int_\mathscr{X}\lambda\!\left(\boldsymbol{f}(\boldsymbol{x})^\top\boldsymbol{\beta}\right)\boldsymbol{f}(\boldsymbol{x})\boldsymbol{f}(\boldsymbol{x})^\top \xi(\D \boldsymbol{x})\ .
	\end{equation*}	
	Here generalized design does not only mean design on a discrete set of design points. It means an arbitrary probability measure on the design region. In contrast a discrete design has a discrete probability measure with discrete or finite support, see, for example, \cite{Silvey:1980}. \\ 	
	So we can define: A design $\xi^\ast$ with regular information matrix $\boldsymbol{M}(\xi^\ast,\boldsymbol{\beta})$ is called (locally) $D$-optimal (at $\boldsymbol{\beta}$) if $\det(\boldsymbol{M}(\xi^\ast,\boldsymbol{\beta}))\geq\det(\boldsymbol{M}(\xi,\boldsymbol{\beta}))$ holds for all possible probability measures $\xi$ on $\mathscr{X}$. 
	
\begin{notation}
	The symbol $\mathbb{S}_{d-1}$, $d\in\{2,3,4,\ldots\}$, describes the unit sphere, which is the surface of a $d$-dimensional unit ball $\mathbb{B}_d$. 
	Introducing notations we also mention $\mathbb{O}_d$ the $d$-dimensional zero-vector, $\mathbb{O}_{d_1\times d_2}$ the $(d_1\times d_2)$-dimensional zero-matrix, $\mathds{1}_d$ the $d$-dimensional one-vector, $\mathbb{I}_d$ the $(d\times d)$-di\-men\-sion\-al identity matrix and $\id$ the identity function.
\end{notation}

	Now we collect some results and lemmas from \cite{Radloff:2018} which will also be valid and helpful for our current endeavour.

\begin{lemma} \label{L2} 
	Any (locally) $D$-optimal design is concentrated on the surface of $\mathscr{X}=\mathbb{B}_k$ and is equivariant with respect to rotations.
\end{lemma}
	Equivariance in this context means: If the design or design region is rotated, the parameter space must be rotated in a corresponding way. For detailed information see \cite{Radloff:2016, Radloff:2018}.\\
	For an initial guess $(\beta_1,\ldots,\beta_k)^\top\neq\mathbb{O}_k$ --- the case $=\mathbb{O}_k$ is discussed later --- there is a rotation $\boldsymbol{\tilde{g}}$ such that $\boldsymbol{\tilde{g}}(\beta_0,\beta_1,\ldots,\beta_k)^\top=(\beta_0,\tilde{\beta}_1,0,\ldots,0)$ with $\tilde{\beta}_1=||(\beta_1,\ldots,\beta_k)^\top||>0$, where $||\cdot||$ is the ($k$-dimensional) Euclidean norm. In view of the equivariance and without loss of generality only the case $\boldsymbol{\beta}\in\mathbb{R}^{k+1}$ with
	\begin{equation}
		\label{eq:speccase}
			\beta_1\geq 0, \beta_2=\ldots=\beta_k=0
	\end{equation}
	has to be considered for optimization. This simplifies our problem of finding a (locally) $D$-optimal design with an initial guess of the parameter vector in the whole parameter space to only the length of this vector.

\begin{lemma} 
\label{lem3}
	For $\boldsymbol{\beta}$ satisfying~\eqref{eq:speccase} the $D$-criterion is invariant with respect to rotations of $x_2,\ldots,x_k$.
\end{lemma}
	So we can find an optimal design within the class of invariant designs on the surface of the ball.

	If the initial guess $(\beta_1,\ldots,\beta_k)^\top$ is $\mathbb{O}_k$ then no rotation $\boldsymbol{\tilde{g}}$ is needed at the beginning and an optimal design is invariant with respect to rotations of all components $x_1,x_2,\ldots,x_k$ because the intensity function $\lambda\!\left(\boldsymbol{f}(\boldsymbol{x})^\top\boldsymbol{\beta}\right)$ is constant in that case. As in the linear model issue the (continuously) uniform design on $\mathbb{S}_{k-1}$ is (locally) $D$-optimal. A $k$-dimensional regular simplex, whose $k+1$ vertices lie on the surface of the design region $\mathbb{S}_{k-1}$, has the same information matrix --- the diagonal matrix $\mathrm{diag}(1,\tfrac{1}{k},\ldots,\tfrac{1}{k})$, see \citet[Section 15.12]{Pukelsheim:1993} or \cite{Radloff:2018}. It can be easily calculated that the vertices of a regular $k$-dimensional cross-polytope ($2\,k$ vertices)  as well as the vertices of a $k$-dimensional cube ($2^k$ vertices) inscribed in the ball $\mathbb{B}_k$ have the same information matrix if equal weights are assigned.   
	
	Note that every design or probability measure on the surface of a unit ball can be split into a marginal probability measure $\xi_1$ on $[-1,1]$ for $x_1$ and a probability kernel given $x_1$. In the case of \eqref{eq:speccase} with $\beta_1>0$ Lemma~\ref{lem4} provides a special property, so that we get the  representations in Lemma~\ref{lem4} for optimal invariant designs, the information matrix and the sensitivity function $$\psi(\boldsymbol{x},\xi_1\otimes\overline{\eta}) = \lambda\!\left(\boldsymbol{f}(\boldsymbol{x})^\top\boldsymbol{\beta}\right)\boldsymbol{f}(\boldsymbol{x})^\top\boldsymbol{M}^{-1}\left(\xi_1\otimes\overline{\eta}\right)\boldsymbol{f}(\boldsymbol{x})$$ which is used in the Kiefer-Wolfowitz Equivalence Theorem for $D$-optimality.

\begin{lemma} \label{lem4} 
	For $\boldsymbol{\beta}$ satisfying~\eqref{eq:speccase} the invariant designs (on the surface) with respect to rotations of $x_2,\ldots,x_k$ are given by $\xi_1\otimes\overline{\eta}$, where $\xi_1$ is a marginal design on $[-1,1]$ and $\overline{\eta}$ is a probability kernel (conditional design). For fixed $x_1$ the kernel $\overline{\eta}(x_1,\cdot)$ is the uniform distribution on the surface of a $(k-1)$-dimensional ball with radius $\sqrt{1-x_1^2}$.\\
	The related information matrix is (remembering $q(x_1)=\lambda(\beta_0+\beta_1 x_1)$)
\renewcommand*\arraystretch{1.2}
	\begin{equation}
			\boldsymbol{M}(\xi_1\otimes\overline{\eta})=
					\left(\begin{array}{c|c}
							\hspace*{-0.5em}\begin{array}{cc}
							\int q\,\D\xi_1 		 & \int q \id \D\xi_1\\
							\int q \id \D\xi_1 & \int q \id^2 \D\xi_1
							\end{array} & \mathbb{O}_{2\times (k-1)}\\ \hline
							\mathbb{O}_{(k-1)\times 2} & \frac{1}{k-1} \int q\,(1-\id^2)\,\D\xi_1\ \mathbb{I}_{k-1}							
					\end{array} \right) .
	\label{eq:infomatrix}
	\end{equation}
\renewcommand*\arraystretch{1}

	\noindent The sensitivity function $\psi$ is invariant (constant on orbits) and has for $\boldsymbol{x}\in\mathbb{S}_{k-1}$ the form
	\begin{equation}
		\psi(\boldsymbol{x},\xi_1\otimes\overline{\eta})=q(x_1)\cdot p_1(x_1) \quad\text{with}\quad \boldsymbol{x}=(x_1,\ldots,x_k)^\top
	\label{eq:5a}
	\end{equation}
	where $p_1$ is a polynomial of degree 2 in $x_1$.	
\end{lemma}
	If $x_1\in\{-1,1\}$, the $(k-1)$-dimensional ball with the uniform distribution is degenerated as a point. So it is only a one-point-measure.

\setcounter{equation}{0} 
\section{Logit and probit model}
\label{sec:3}	
	The intensity function for the logit model is
	\[\lambda_{\mathrm{logit}}(x)=\frac{\exp(x)}{(1+\exp(x))^2}\]
	and for the probit model
	\[\lambda_{\mathrm{probit}}(x)=\frac{\phi^2(x)}{\Phi(x)(1-\Phi(x))}\]
	with the density function $\phi$ and cumulative distribution function $\Phi$ of the standard normal distribution.
	
	As mentioned before the intensity function of the binary response models with logit or probit link do not satisfy the conditions (A2) and (A3). But they satisfy 
	\begin{enumerate}[leftmargin=1cm]
		\item[(A2$^\prime$)] $\lambda$ is unimodal with mode $c_\lambda\in\mathbb{R}$, which means that there exists a $c_\lambda\in\mathbb{R}$ so that $\lambda^\prime$ is positive on $(-\infty,c_\lambda)$ and negative on $(c_\lambda,\infty)$.
		\item[(A3$^\prime$)] There exists a $c_\lambda\in\mathbb{R}$ so that the second derivative $u^{\prime\prime}$ of $u=\frac{1}{\lambda}$ is both injective on $(-\infty,c_\lambda]$ and injective on $[c_\lambda,\infty)$.
	\end{enumerate}
	If (A2$^\prime$) and (A3$^\prime$) are fulfilled it should be the same $c_\lambda$.
	As the properties (A1)-(A4) transfer from the intensity function $\lambda$ to the abbreviated form $q$ for $\beta_1>0$ and vice versa, the same is to (A2$^\prime$) and (A3$^\prime$) --- analogously $c_q=\frac{c_\lambda-\beta_0}{\beta_1}$.
	
	It is 
	\begin{align*}
		q_\mathrm{logit}^\prime(x_1)&=\beta_1\,\frac{\exp(\beta_0+\beta_1 x_1)\,(1-\exp(\beta_0+\beta_1 x_1))}{(1+\exp(\beta_0+\beta_1 x_1))^3}\\
		u_\mathrm{logit}(x_1)
				&=2+\exp(\beta_0+\beta_1 x_1)+\exp(-(\beta_0+\beta_1 x_1))\\
		u_\mathrm{logit}^{\prime\prime}(x_1)
				&=\beta_1^2\,\left(\exp(\beta_0+\beta_1 x_1)+\exp(-(\beta_0+\beta_1 x_1))\right)
	\end{align*}
	in logit model. Without writing down the terms of the probit model here we have in both models $c_\lambda=0$ for $\lambda$ and the analogue $c_q=-\frac{\beta_0}{\beta_1}$ for $q$.
	
	We introduce a fifth property.
	\begin{enumerate}[leftmargin=1cm]
		\item[(A5)] $u=\frac{1}{\lambda}$ dominates $x^2$ asymptotically for $x\to\infty$, 
								which means \[\lim\limits_{x\to\infty}\left|\frac{u(x)}{x^2}\right|=\infty.\]
	\end{enumerate}
	In other words $u(x)=\frac{1}{\lambda(x)}$ goes faster to ($\pm$) infinity than $x^2$ for $x\to\infty$.
	The logit and probit models satisfy (A5).

	\begin{lemma} \label{L6} 
		In~\eqref{eq:speccase}: If $q$ satisfies (A1), (A2\,$^\prime$) and (A3\,$^\prime$), then the (locally) $D$-optimal marginal design $\xi_1^\ast$ is concentrated on exactly 2 points $x_{11}^\ast, x_{12}^\ast\in[-1,1]$ or exactly 3 points $x_{11}^\ast=1$, $x_{12}^\ast\in(-1,1)$ and $x_{13}^\ast=-1$.\\
		If $q$ satisfies additionally (A5) then only the 2-point structure is possible.
	\end{lemma}

	\begin{proof}
		This proof is based on the proof of Lemma~1 in \cite{Konstantinou:2014}.
		By the Kiefer-Wolfowitz Equivalence Theorem for $D$-optimality we have to check 
		\[k+1\geq\psi(\boldsymbol{x},\xi_1\otimes\overline{\eta})=q(x_1)\cdot p_1(x_1)\quad\text{for all}\quad\boldsymbol{x}=(x_1,\ldots,x_k)^\top .\]
		This is equivalent to 
		\begin{equation}
			\frac{p_1(x_1)}{k+1} - \frac{1}{q(x_1)} \leq 0\ .
		\label{eq:IneqEquivalenceThBiedermann}
		\end{equation}
		With equality in the support points of the optimal design.\\
		Assume, that $\xi_1$ has only 1 support point. So the determinant of the first block of the information matrix $\boldsymbol{M}(\xi_1\otimes\overline{\eta})$ in Lemma~\ref{lem4} would be 0 and the inverse of the information matrix and thus the polynomial $p_1$ would not exist. Contradiction. Hence, $\xi_1$ has at least 2~support points.\\
		Let us call the left-hand side of \eqref{eq:IneqEquivalenceThBiedermann} $v(x_1)$. The second derivative of $v$ is $v^{\prime\prime}(x_1)=\tilde{c}-\left(\frac{1}{q(x_1)}\right)^{\prime\prime}$ where $\tilde{c}$ is the constant remaining from the polynomial $\frac{p_1(x_1)}{k+1}$ of degree~2 (see Lemma~\ref{lem4}). The condition (A3$^\prime$) says that $v^{\prime\prime}$ can have at most 2~roots. Because of differentiability and continuity the first derivative of  $v$ has at most 3~roots which means that $v$ has at most 3~potential inner local extreme points with alternating minima and maxima. If it is minimum-maximum-minimum then $x_{11}^\ast=1$, $x_{12}^\ast\in(-1,1)$ and $x_{13}^\ast=-1$ can be the 3~maxima of $v$ since 1 and $-1$ are boundary points. If additionaly (A5) is satisfied, $\lim_{x_1\to\infty}v(x_1)=-\infty$ so that $1$ cannot be a boundary maximum if the other 3~local extreme points are less than 1. In the case of (A5) the only situation with exactly 3 inner extreme points is maximum-minimum-maximum. In all other cases there are at most 2~maxima (inner or boundary) and so at most 2~support points.
	\end{proof}

The next lemma characterizes the support points when the design has exactly 2.

\begin{theorem}\label{T2}
	In the settings of Lemma~\ref{L6} and with~$q$ satisfying (A5) the (locally) $D$-optimal marginal design~$\xi_1^\ast$ has exactly 2~support points $x_{11}^\ast$ and $x_{12}^\ast$ with $x_{11}^\ast>x_{12}^\ast$ and weights $w_1:=\xi_1^\ast(x_{11}^\ast)$ and $w_2:=\xi_1^\ast(x_{12}^\ast)$.\\
		There are 3~cases:
	\begin{enumerate}
		\item[a)] $\frac{c_\lambda-\beta_0}{\beta_1}  > 1$:     	$x_{11}^\ast=1, \quad w_1=\frac{1}{k+1}, \quad w_2=\frac{k}{k+1}$\\
																									For $k\geq 2\ :\ x_{12}^\ast\in(-1,1)$ is solution of 
																									\begin{equation*}
																										\frac{q^\prime(x_{12}^\ast)}{q(x_{12}^\ast)}=\frac{2\,(1+kx_{12}^\ast)}{k\,(1-x_{12}^{\ast\ 2})}
																									\end{equation*}
																									and for $k=1$\ :\ If $x$ is solution of
																									\begin{equation*}
																										\frac{q^\prime(x)}{q(x)}=\frac{2}{1-x}
																									\end{equation*}
																									and $x\in[-1,1)$ then $x_{12}^\ast=x$. Otherwise $x_{12}^\ast=-1$.\\																							
		
																									In any case, if additionally (A4) is satisfied, the solution $x_{12}^\ast$ is unique.\\
		\item[b)] $\frac{c_\lambda-\beta_0}{\beta_1}  < -1$:      $x_{12}^\ast=-1, \quad w_1=\frac{k}{k+1}, \quad w_2=\frac{1}{k+1}$\\
																									For $k\geq 2\ :\ x_{11}^\ast\in(-1,1)$ is solution of 
																									\begin{equation*}
																										\frac{q^\prime(x_{11}^\ast)}{q(x_{11}^\ast)}=\frac{2\,(-1+kx_{11}^\ast)}{k\,(1-x_{11}^{\ast\ 2})}
																									\end{equation*}
																									and for $k=1$\ :\ If $x$ is solution of
																									\begin{equation*}
																										\frac{q^\prime(x)}{q(x)}=\frac{-2}{1+x}
																									\end{equation*}
																									and $x\in(-1,1]$ then $x_{11}^\ast=x$. Otherwise $x_{11}^\ast=1$.\\		
																									
																									In any case, if additionally (A4) is satisfied, the solution $x_{11}^\ast$ is unique.\\ 
		\item[c)] $\frac{c_\lambda-\beta_0}{\beta_1}\in [-1,1]$: If $x,y\in(-1,1)$ with $x>y$ and $\alpha\in(0,1)$ is a solution of the equation system
				\begin{align*} 		
					\frac{q^\prime(x)}{q(x)}+\frac{2}{x-y}+(k-1)\,\frac{q^\prime(x)\,(1-x^2)\,\alpha + q(x)\,(-2\,x)\,\alpha}{q(x)\,(1-x^2)\,\alpha+q(y)\,(1-y^2)\,(1-\alpha)}&=0\\
					\frac{q^\prime(y)}{q(y)}-\frac{2}{x-y}+(k-1)\,\frac{q^\prime(y)\,(1-y^2)\,(1-\alpha) + q(y)\,(-2\,y)\,(1-\alpha)}{q(x)\,(1-x^2)\,\alpha+q(y)\,(1-y^2)\,(1-\alpha)}&=0\\
					\frac{1}{\alpha}-\frac{1}{1-\alpha}+(k-1)\,\frac{q(x)\,(1-x^2) - q(y)\,(1-y^2)}{q(x)\,(1-x^2)\,\alpha+q(y)\,(1-y^2)\,(1-\alpha)}&=0
				\label{eq:}
				\end{align*}		
				then the 2~support points are $x_{11}^\ast=x$, $x_{12}^\ast=y$ with weights $w_1=\alpha$ and $w_2=1-\alpha$. Otherwise the solution is in the form of the first two cases.		
	\end{enumerate}
\end{theorem}
	
\begin{proof}
	In \textit{a)} for all $x_1\in[-1,1]$ (A1), (A2) and (A3) are satisfied. And this is the situation of Theorem~\ref{Theorem1}.\\
	In \textit{b)} for all $x_1\in[-1,1]$ (A1) and (A3) are satisfied, but $\lambda$ or $q$, respectively, are strictly decreasing. Using the reflection $x_1\mapsto -x_1$ (A2) is also on hand. Equivariance yields that the optimal design of Theorem~\ref{Theorem1} has to be reflected, too.\\
	According to \cite{Radloff:2018} we know the logarithmized determinant of the information matrix $\boldsymbol{M}(\xi_1\otimes\overline{\eta})$ with a 2-point marginal design
	\begin{align*} 
		&\log q(x_{11}^\ast) + \log q(x_{12}^\ast) + \log (x_{11}^\ast-x_{12}^\ast)^2 + \log \alpha + \log(1-\alpha)\\ &+ (k-1)\left[-\log(k-1)+\log\left(q(x_{11}^\ast)\,(1-x_{11}^{\ast\ 2})\,\alpha+q(x_{12}^\ast)\,(1-x_{12}^{\ast\ 2})\,(1-\alpha)\right)\right]
	\end{align*} 		
	which has to be maximized in \textit{c)}. If $x_{11}^\ast, x_{12}^\ast\notin(-1,1)$ and $\alpha\notin(0,1)$ then there must be a boundary maximum. If one point is fixed to 1 or $-1$ we get the same situation as in \textit{a)} or \textit{b)}, respectively.
\end{proof}

	\begin{figure} [htb!]
		\raisebox{0.42\textwidth}{\parbox[t][][t]{4em}{$w_1$\\[7.7ex]$w_2$\\[0.7ex]$x_{11}^\ast$\\[18.5ex]$x_{12}^\ast$}}\hspace*{-4em}
		\subfigure[$k=3$]{\includegraphics[width=0.475\textwidth]{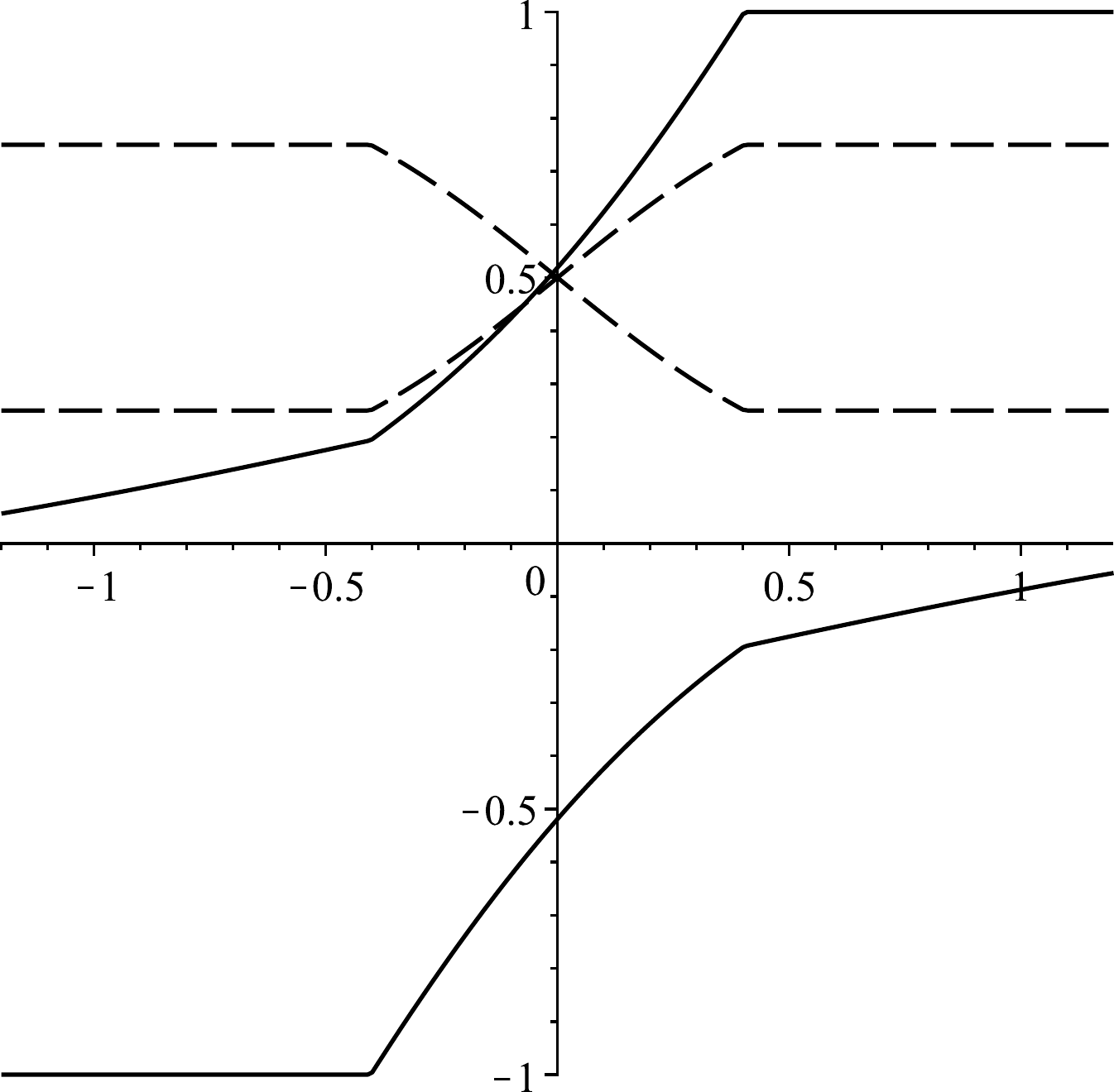}}
		\hspace*{-3em}\raisebox{0.17\textwidth}{$-\beta_0$}
    \hspace*{0.025\textwidth}
		\raisebox{0.44\textwidth}{\parbox[t][][t]{4em}{$w_1$\\[11.5ex]$w_2$\\[5ex]$x_{11}^\ast$\\[12ex]$x_{12}^\ast$}}\hspace*{-4em}
    \subfigure[$k=6$]{\includegraphics[width=0.475\textwidth]{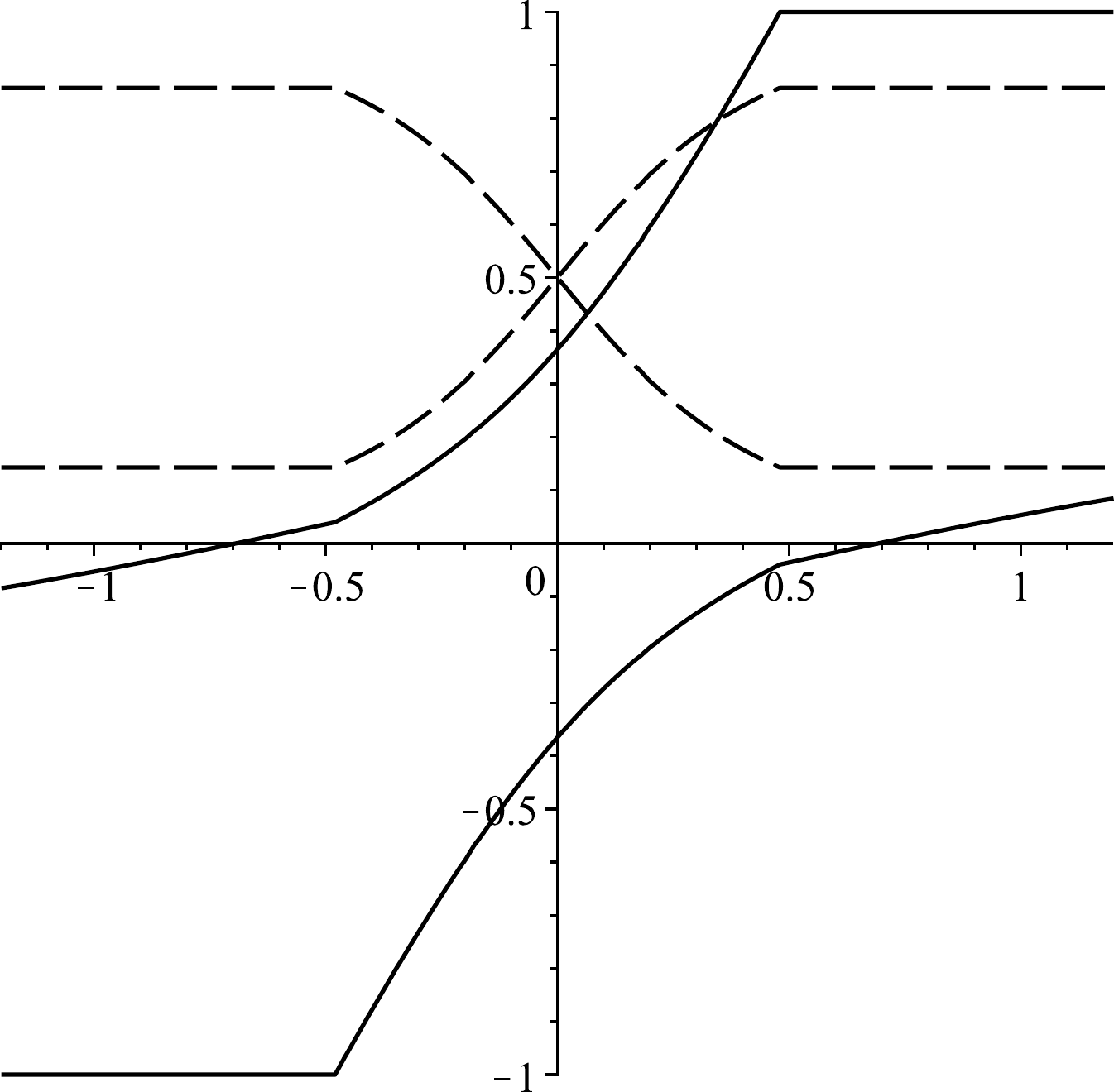}}
		\hspace*{-3em}\raisebox{0.17\textwidth}{$-\beta_0$}
\caption{Logit model: Dependence of $x_{11}^\ast$ and $x_{12}^\ast$ (solid lines) and the corresponding weights $w_1$ and $w_2=1-w_1$ (dashed lines) on $-\beta_0\in[-1.2,1.2]$. The plots are for fixed dimension~$k$ and $\beta_1=1$. Hence, $-\beta_0=-\frac{\beta_0}{\beta_1}=c_q$.}
	\label{fig:plot_logit_k3_k6}
\end{figure}
	
	Using Theorem~\ref{T2} we can evaluate the two support points of the marginal design $\xi_1$ of the logit model. In Figure~\ref{fig:plot_logit_k3_k6} we did this numerically for $\beta_0\in[-1.2,1.2]$, fixed $\beta_1=1$ and the dimensions $k=3$ or $k=6$. The situation \textit{c)} where we get two real inner points is only for $\beta_0\in(-0.403,0.403)$ (approximated) for $k=3$ and $\beta_0\in(-0.480,0.480)$ for $k=6$. In the probit model the plots in Figure~\ref{fig:plot_probit_k3_k6} have nearly the same structure. The interesting part, where we have two inner points, here is in $(-0.436,0.436)$ for $k=3$ and in $(-0.507,0.507)$ for $k=6$. 

	\begin{figure} [htb!]
		\raisebox{0.42\textwidth}{\parbox[t][][t]{4em}{$w_1$\\[7.7ex]$w_2$\\[7.5ex]$x_{11}^\ast$\\[11.5ex]$x_{12}^\ast$}}\hspace*{-4em}
		\subfigure[$k=3$]{\includegraphics[width=0.475\textwidth]{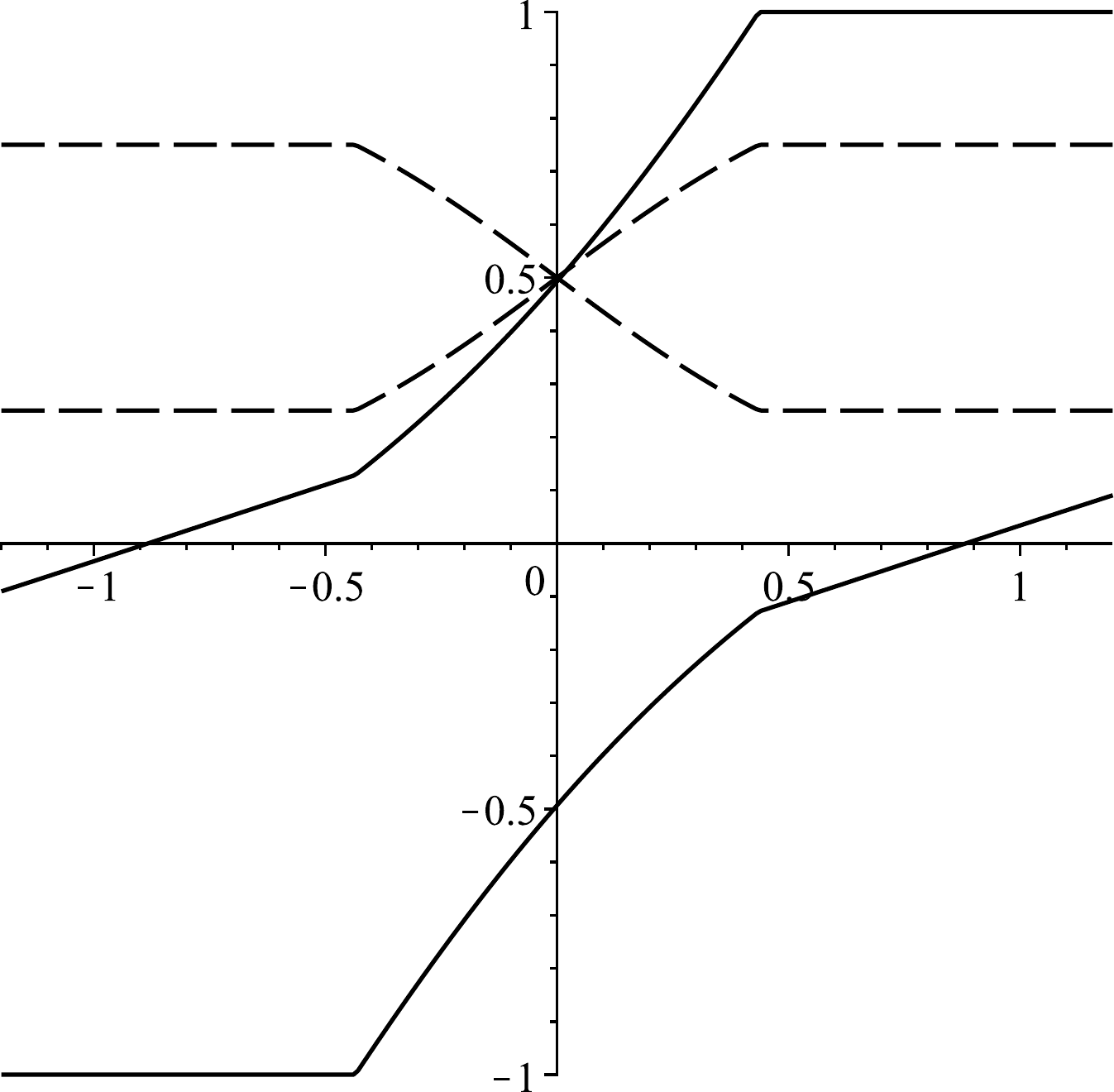}}
		\hspace*{-3em}\raisebox{0.17\textwidth}{$-\beta_0$}
    \hspace*{0.025\textwidth}
		\raisebox{0.44\textwidth}{\parbox[t][][t]{4em}{$w_1$\\[11.5ex]$w_2$\\[8ex]$x_{11}^\ast$\\[9ex]$x_{12}^\ast$}}\hspace*{-4em}
    \subfigure[$k=6$]{\includegraphics[width=0.475\textwidth]{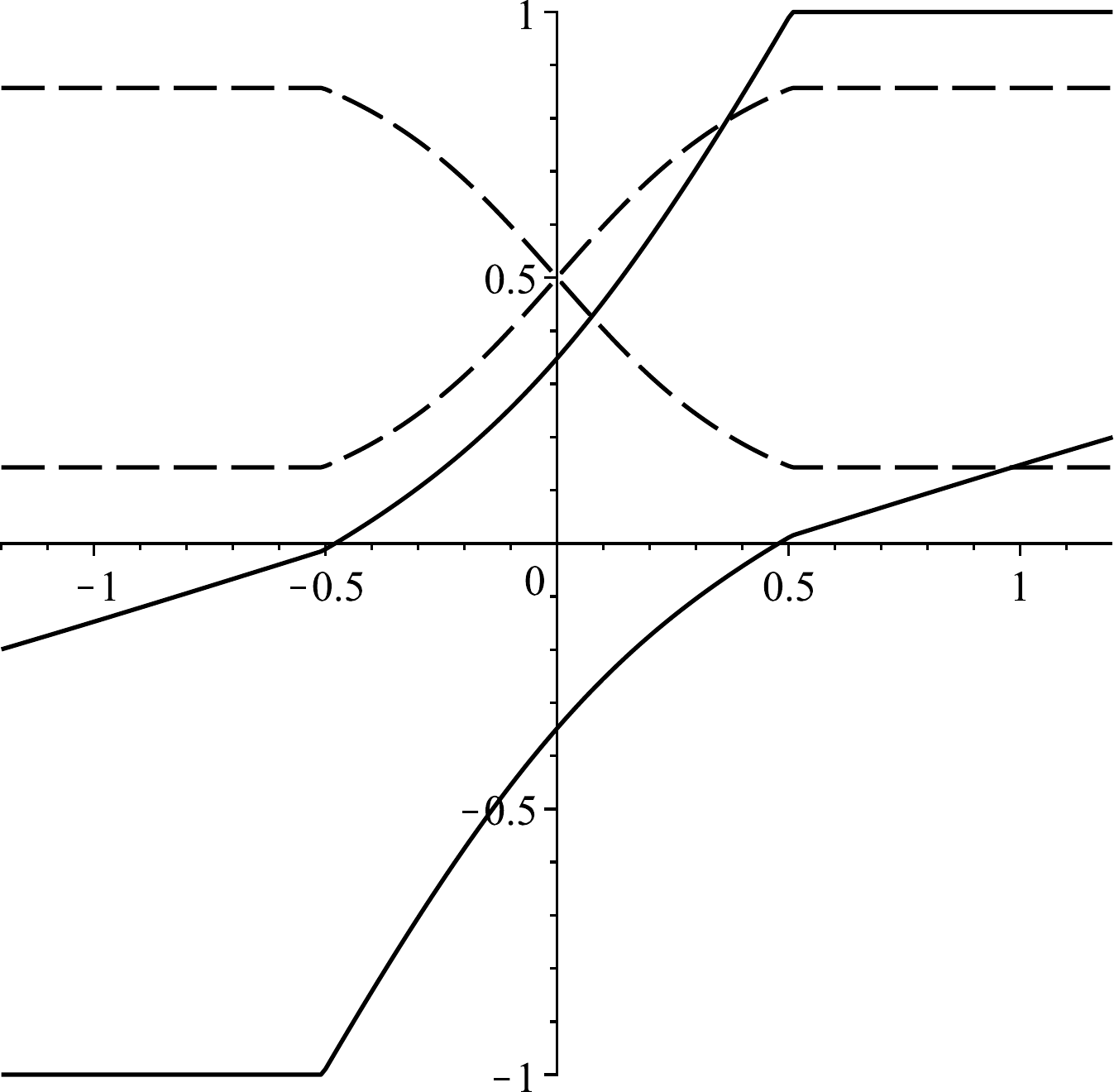}}
		\hspace*{-3em}\raisebox{0.17\textwidth}{$-\beta_0$}
\caption{Probit model: Dependence of $x_{11}^\ast$ and $x_{12}^\ast$ (solid lines) and the corresponding weights $w_1$ and $w_2=1-w_1$ (dashed lines) on $-\beta_0\in[-1.2,1.2]$. The plots are for fixed dimension~$k$ and $\beta_1=1$. Hence, $-\beta_0=-\frac{\beta_0}{\beta_1}=c_q$.}
	\label{fig:plot_probit_k3_k6}
\end{figure}

	But there is a big difference which cannot be seen in the two figures. The behaviour of the inner point for $-\beta_0\to\infty$ or $-\beta_0\to-\infty$ and arbitrary $\beta_1\geq 0$. In the probit model the inner point converges from below to 1 or from above to $-1$, respectively. In the logit model the inner point converges to  
	\begin{equation*}
		\begin{cases}
			\frac{-1+\sqrt{1-\frac{2}{k}\beta_1+\beta_1^2}}{\beta_1} & \text{for $\beta_1>0$}\\
			-\frac{1}{k}& \text{for $\beta_1=0$\ .}
		\end{cases} \notag
	\end{equation*}
	For $\beta_1=1$ we get $-1+\sqrt{\frac{4}{3}}\approx0.1547$ ($k=3$) and $-1+\sqrt{\frac{5}{3}}\approx0.2910$ ($k=6$).

	As in Theorem~\ref{Theorem1} the orbit belonging to the inner point $x_{12}^\ast$ or $x_{11}^\ast$ in situation \textit{a)} or \textit{b)} can be discretized by the vertices of a $(k-1)$-dimensional regular simplex. So we have exact (locally) $D$-optimal designs with equal weights $\frac{1}{k+1}$.\\
	The discretization in \textit{c)} is more difficult. If the weight $w_1$ and $w_2$ are appropriated it can be done as mentioned above by using $(k-1)$-dimensional regular simplices, cross-polytopes, cubes or combinations of them.

	\begin{figure} [htb!]
		\subfigure[$\beta_0=0$]{\includegraphics[width=0.3\textwidth]{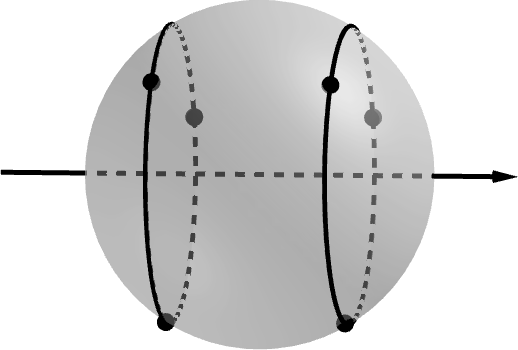}}
    \hspace*{0.02\textwidth}
    \subfigure[$\beta_0=-0.5$]{\includegraphics[width=0.3\textwidth]{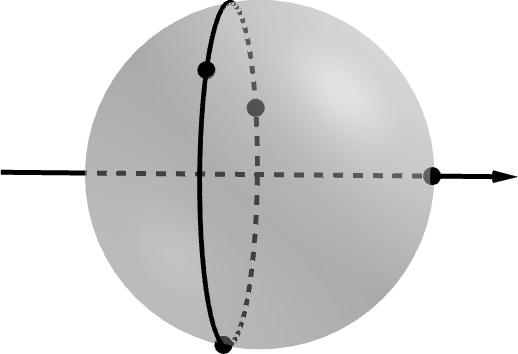}}
    \hspace*{0.02\textwidth}
		\subfigure[$\beta_0=0.1$]{\includegraphics[width=0.3\textwidth]{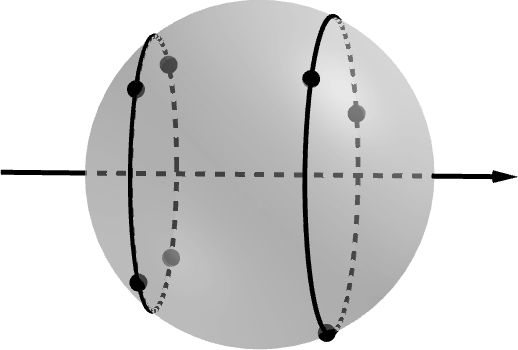}}
\caption{Logit model: Discretized (locally) $D$-optimal designs for $\beta_1=1$ and $k=3$.}
	\label{fig:logit_k3_beta0_0_-0-5}
\end{figure}

As three examples we want to focus the logit model with $\beta_1=1$ and $k=3$, see Figure~\ref{fig:logit_k3_beta0_0_-0-5}.
For $\beta_0=-0.5$ we get $x_{11}^\ast=1$, $x_{12}^\ast\approx-0.18$ and $w_1=\xi_1^\ast(x_{11}^\ast)=\frac{1}{4}$.
For $\beta_0=0$ we get apart from rotation invariance with respect to $x_2,\ldots,x_k$ an extra invariance --- the reflection in $x_1$-direction. In addition the intensity function of the logit model is symmetrical. Therefore the two support points of the marginal design must be symmetrical around 0, that is $x_{11}^\ast=-x_{12}^\ast$, and the weights must be equal $\xi_1^\ast(x_{11}^\ast)=\xi_1^\ast(x_{12}^\ast)=0.5$. By calculation we get $x_{11}^\ast=-x_{12}^\ast\approx0.52$.  So both designs have equal weights on their support points. While the optimal design for $\beta_0=-0.5$ has the minimum number of points the optimal design for $\beta_0=0$ consists of two 2-dimensional simplices. So it may be possible that there is another optimal design with less than 6 support points.\\
In case of $\beta_0=0.1$ we get $x_{11}^\ast\approx0.42$, $x_{12}^\ast\approx-0.62$ and $\xi_1^\ast(x_{11}^\ast)\approx0.4297\approx\frac{3}{7}$. So we decided to substitute one orbit by the vertices of a 2-dimensional simplex (3~points) and one by the vertices of a 2-dimensional cube or cross polytope, which is in two dimensions always a square (4~points). Okay there is a little bit rounding, but it is near to the optimum. To verify this we can calculate the $D$-efficiency which compares the rounded design $\xi_\approx$ and the (non-rounded) optimal design $\xi^\ast$:
\[\mathrm{eff}_D(\xi_\approx):=\left(\frac{\det\boldsymbol{M}(\xi_\approx)}{\det\boldsymbol{M}(\xi^\ast)}\right)^\frac{1}{k}\approx0.999676\ .\]

\setcounter{equation}{0} 
\section{Summary and Discussion}
\label{sec:4}	

In the present paper we developed (locally) $D$-optimal designs for a class of non-linear multiple regression problems which include especially binary response models with logit or probit link. This extension of the results established in \cite{Radloff:2018} provides in certain cases exact designs. In all other cases rotation-invariant approximate designs are obtained which consist of two parallel (non-degenerate) orbits on the surface of the spherical design region of a $k$-dimensional ball.

For practical applications one may imagine problems in engineering or physics where the validity of a model may be assumed on a spherical region around a target value, for example in the framework of response surface methodology.


By using linear transformations, like scaling and rotating, the class of shapes of the design region can be extended from the unit ball to $k$-dimensional balls with arbitrary radius or any $k$-dimensional ellipsoid, which can be obtained by using the equivariance results established in \cite{Radloff:2016}.
 
Here we focused on linear regressors of the multiple linear regression type. Accounting for interactions or quadratic terms will presumably induce additional support points in the interior of the design region and/or more complicated design structures.

There is one property observed in the numerical calculations for both the logit and probit model (see Figures~\ref{fig:plot_logit_k3_k6} and~\ref{fig:plot_probit_k3_k6}) which deserves further investigations: If the intensity function $\lambda$ is symmetrical, that means $\lambda(c_\lambda+x)=\lambda(c_\lambda-x)$, then the two support points are also symmetric around $c_\lambda$ as long as these support points are in the interior of the marginal design region. We observed this in the case of logit and probit models, see Figures~\ref{fig:plot_logit_k3_k6} and~\ref{fig:plot_probit_k3_k6}. For the one-dimensional case this has been proved in \citet[Section 6.5 and 6.6]{Ford:1992}, but this proof cannot be extended to higher dimensions directly because of the additional asymmetric term $(1-x_1^2)$.

As in \cite{Radloff:2018} we only considered the criterion of (local) $D$-optimality which depends on the actual value of the parameter vector. In general other optimality criteria or especially (more) robust criteria, like maxi\-min efficiency or weighted criteria, should be the object of future research also in the present context.


\bibhang=1.7pc
\bibsep=2pt
\fontsize{9}{14pt plus.8pt minus .6pt}\selectfont
\renewcommand\bibname{\large \bf References}
\expandafter\ifx\csname
natexlab\endcsname\relax\def\natexlab#1{#1}\fi
\expandafter\ifx\csname url\endcsname\relax
  \def\url#1{\texttt{#1}}\fi
\expandafter\ifx\csname urlprefix\endcsname\relax\def\urlprefix{URL}\fi

\bibliographystyle{chicago}      
\bibliography{BibTeX_Radloff_1}   

\begin{thebibliography}{}

\bibitem[\protect\citeauthoryear{Atkinson, Fedorov, Herzberg, and
  Zhang}{Atkinson et~al.}{2014}]{Atkinson:2014}
Atkinson, A.~C., V.~V. Fedorov, A.~M. Herzberg, and R.~Zhang (2014).
\newblock Elemental information matrices and optimal experimental design for
  generalized regression models.
\newblock {\em Journal of Statistical Planning and Inference\/}~{\em 144},
  81--91.

\bibitem[\protect\citeauthoryear{Biedermann, Dette, and Zhu}{Biedermann
  et~al.}{2006}]{Biedermann:2006}
Biedermann, S., H.~Dette, and W.~Zhu (2006, June).
\newblock Optimal designs for dose-response models with restricted design
  spaces.
\newblock {\em Journal of the American Statistical Association\/}~{\em
  101\/}(474), 747--759.

\bibitem[\protect\citeauthoryear{Fedorov}{Fedorov}{1972}]{Fedorov:1972}
Fedorov, V.~V. (1972).
\newblock {\em Theory of optimal experiments}.
\newblock Academic Press.

\bibitem[\protect\citeauthoryear{Ford, Torsney, and Wu}{Ford
  et~al.}{1992}]{Ford:1992}
Ford, I., B.~Torsney, and C.~Wu (1992).
\newblock The use of a canonical form in the construction of locally optimal
  designs for non-linear problems.
\newblock {\em Journal of the Royal Statistical Society: Series B (Statistical
  Methodology)\/}~{\em 54\/}(2), 569--583.

\bibitem[\protect\citeauthoryear{Konstantinou, Biedermann, and
  Kimber}{Konstantinou et~al.}{2014}]{Konstantinou:2014}
Konstantinou, M., S.~Biedermann, and A.~Kimber (2014).
\newblock Optimal designs for two-parameter nonlinear models with application
  to survival models.
\newblock {\em Statistica Sinica\/}~{\em 24\/}(1), 415--428.

\bibitem[\protect\citeauthoryear{Pukelsheim}{Pukelsheim}{1993}]{Pukelsheim:1993}
Pukelsheim, F. (1993).
\newblock {\em Optimal design of experiments}.
\newblock Wiley Series in Probability and Statistics.

\bibitem[\protect\citeauthoryear{Radloff and Schwabe}{Radloff and
  Schwabe}{2016}]{Radloff:2016}
Radloff, M. and R.~Schwabe (2016).
\newblock Invariance and equivariance in experimental design for nonlinear
  models.
\newblock In J.~Kunert, C.~H. M{\"u}ller, and A.~C. Atkinson (Eds.), {\em mODa
  11-Advances in Model-Oriented Design and Analysis}, pp.\  217--224. Springer.

\bibitem[\protect\citeauthoryear{Radloff and Schwabe}{Radloff and
  Schwabe}{2018}]{Radloff:2018}
Radloff, M. and R.~Schwabe (2018).
\newblock Locally $d$-optimal designs for non-linear models on the
  $k$-dimensional ball.

\bibitem[\protect\citeauthoryear{Schmidt and Schwabe}{Schmidt and
  Schwabe}{2017}]{Schmidt:2017}
Schmidt, D. and R.~Schwabe (2017).
\newblock Optimal design for multiple regression with information driven by the
  linear predictor.
\newblock {\em Statistica Sinica\/}~{\em 27\/}(3), 1371--1384.

\bibitem[\protect\citeauthoryear{Silvey}{Silvey}{1980}]{Silvey:1980}
Silvey, S.~D. (1980).
\newblock {\em Optimal design: an introduction to the theory for parameter
  estimation}.
\newblock Chapman and Hall.

\end{thebibliography}

\end{document}